\documentclass[times,5p,sort&compress]{elsarticle}
\pdfoutput=1
\usepackage[utf8]{inputenc}
\usepackage[T1]{fontenc}
\usepackage[english]{babel}
\usepackage{amsmath,amssymb,amsthm}
\usepackage{hyphenat}
\usepackage[textsize=footnotesize]{todonotes}
\usepackage{microtype}
\usepackage{subcaption}
\usepackage{lipsum}
\usepackage{mathtools}
\usepackage{color,soul}

\usepackage{hyperref}
\usepackage{doi}
\usepackage[sort&compress,nameinlink,noabbrev,capitalize]{cleveref}
\usepackage{paralist}
\usepackage{tabularx}
\usepackage{booktabs}
\hypersetup{colorlinks,allcolors=blue}
\newtheoremstyle{iostuff}%
{0pt}%
{0pt}%
{\hangindent=\parindent}%
{}%
{\itshape}%
{:}%
{.5em}%
{}%

\numberwithin{equation}{section}
\numberwithin{figure}{section}
\newcommand{\conc}{\circ}
\newcommand{\hopg}{H}
\newcommand{\hopgt}{hop graph}
\newcommand{\tour}{\ensuremath{P}}
\newcommand{\vc}{\tau}

\newcommand{\N}{\ensuremath{\mathbb N}}
\newcommand{\cost}{\ensuremath{\omega}}
\newcommand{\w}{\ensuremath{\cost}}
\newcommand{\wh}{\ensuremath{\cost^H}}

\newcommand{\gtsp}{GTSP}
\newcommand{\gtsplong}{Graphical TSP}

\newcommand{\fS}{[\mathcal{S}]}
\theoremstyle{definition}
\newtheorem{theorem}{Theorem}[section]

\newtheorem{problem}[theorem]{Problem}
\newtheorem{lemma}[theorem]{Lemma}

\newtheorem{definition}[theorem]{Definition}

\newtheorem{remark}[theorem]{Remark}
\newtheorem{rrule}[theorem]{Reduction Rule}

\theoremstyle{iostuff}
\newtheorem*{probinstance}{Input}

\newtheorem*{probquest}{Question}

\crefname{rrule}{Reduction Rule}{Reduction Rules}
\crefname{problem}{Problem}{Problems}
\crefname{observation}{Observation}{Observations}
\crefname{construction}{Construction}{Constructions}
\crefname{figure}{Figure}{Figures}
\crefname{conjecture}{Conjecture}{Conjectures}

\date{}

\begin{document}
\begin{frontmatter}
  \title{A quadratic\hyp order problem kernel for the traveling salesman problem
    parameterized by the vertex cover number\tnoteref{t1}}
  \author[hw]{René van Bevern}
  \ead{rene.van.bevern@huawei.com}

  \address[hw]{Huawei Technologies Co., Ltd., Novosibirsk, Russian Federation}
  \tnotetext[t1]{The results presented in this work were first presented in 2021 in the second author's Bachelor thesis
    \citep{Ska21} and are unrelated to the first author's work at Huawei.}
  \author[mipt]{Daniel A.\ Skachkov}
  \ead{skachkov.da@phystech.edu}

  \address[mipt]{Moscow Institute of Physics and Technology, Moscow, Russian Federation}
  
  \begin{abstract}
    The NP\hyp hard
    graphical traveling salesman problem (GTSP)
    is
    to find a closed walk of total minimum weight
    that visits each vertex in
    an undirected edge\hyp weighted and
    not necessarily complete graph.
    We present a problem kernel with $\vc^2+\vc$ vertices for GTSP,
    where $\vc$ is the vertex cover number of the input graph.
    Any $\alpha$\hyp approximate solution
    for the problem kernel also gives an $\alpha$\hyp approximate solution
    for the original instance,
    for any $\alpha\geq 1$.
  \end{abstract}
\end{frontmatter}

\thispagestyle{empty}

\section{Introduction}

\noindent
\citet{BCK+22} have recently studied data reduction with performance guarantees for the
following NP\hyp hard
variant of the traveling salesman problem (TSP).
\begin{problem}[\gtsplong{} (\gtsp)]
  \begin{probinstance}
    An undirected graph~$G=(V,E)$ with edge weights $\cost\colon E\to\N$, and~$W\in \N$.
  \end{probinstance}
  \begin{probquest}
    Find a closed walk containing all vertices of~$V$
    at least once
    and having total edge weight at most~$W$.
  \end{probquest}
\end{problem}

\noindent
Among other results,
\citet{BCK+22} have shown that
any \gtsp{} instance
can be polynomial\hyp time
reduced to an equivalent instance
with $O(\vc^3)$~vertices,
$O(\vc^4)$~edges,
and total bit\hyp size $O(\vc^{16})$,
where $\vc$~is the vertex cover number of~$G$.
In terms of parameterized complexity theory,
they have shown a \emph{problem kernel} \citep{FLSZ19}.
In this work,
we show a smaller problem kernel:

\begin{theorem}\label{thm}
  \gtsp{} admits a problem kernel with
  \begin{enumerate}[(i)]
  \item\label{thm:1} at most   $\vc^2+\vc$ vertices,
    
  \item\label{thm:2} at most $2\vc^3-\vc$ edges,
    
  \item\label{thm:3} total bit\hyp size $O(\vc^{12})$.
  \end{enumerate}
\end{theorem}
\begin{remark}
  \label[remark]{rem:1apx}
  In terms of \citet{LPRS17},
  our problem kernel is \emph{1\hyp approximate};
  that is,
  for any $\alpha\geq1$,
  any $\alpha$\hyp approximate
  solution for the problem kernel
  can be lifted to an $\alpha$\hyp approximate
  solution for the input \gtsp{} instance in polynomial time.
\end{remark}
\noindent
For proving \cref{thm},
we first introduce some notation in \cref{sec:prelims}.
Then,
we present our data reduction rule in \cref{sec:dr},
bound its result size in \cref{sec:size},
and,
finally,
prove its correctness in \cref{sec:correct}.

\section{Preliminaries}
\label{sec:prelims}
\noindent
For a set~$V$, we denote
by $(V)_2 := (V \times V) \setminus \{(v, v) \mid v \in V\}$ the set of ordered pairs of distinct elements.

\paragraph{Undirected graphs}
We consider simple
undirected graphs~$G=(V,E)$
with a set~$V(G):=V$ of \emph{vertices}
and
a set~$E(G):=E\subseteq \{\{u,v\}\subseteq V\mid u\ne v\}$
of \emph{edges}.
Unless stated otherwise,
$n$~denotes the number
of vertices (its \emph{order})
and $m$~denotes the number of edges.
By $N(v):=\{u\in V\mid \{u,v\}\in E\}$,
we denote the \emph{(open) neighborhood} of~$v$.
A \emph{vertex cover}~$C\subseteq V$
is a subset of vertices
such that each edge has at least one endpoint in~$C$.
The minimum cardinality of any vertex cover is denoted
by~$\vc$.
A~matching~$M\subseteq E$ is a set of
edges that do not share endpoints.
  
\paragraph{Walks and tours}
A \emph{walk from vertex~$v_0$ to vertex~$v_\ell$}
in a graph~$G=(V,E)$ is a
sequence~$w=(v_0,v_1,\allowbreak\dots,v_\ell)$
such that $\{v_{i-1},v_{i}\}\in E$
for each~$i\in\{1,\dots,\ell\}$.
Vertices on a walk may repeat.
The \emph{weight} of walk~$w$
is $\w(w):=\sum_{i=1}^\ell\w(\{v_{i-1},v_i\})$.
If $v_0=v_\ell$,
then we call~$w$ a \emph{closed walk}.
A~\emph{subwalk~$w'$} of~$w$ is
any subsequence~$w'$ of~$w$.
The \emph{reversed walk}~$w^R$ of~$w$
is obtained simply by reversing
the sequence of vertices in~$w$.
For two walks~$w_1=(v_0,\dots,v_\ell)$
and $w_2=(v_0',\dots,v'_k)$ such that
$v_\ell=v_0'$,
we denote their concatenation by
$w_1\conc w_2:=(v_0,\dots,v_\ell=v_0',\dots,v'_k)$.

\paragraph{TSP tours}
We call a closed walk containing
all vertices of a graph~$G=(V,E)$
a \emph{TSP tour}.
We call a TSP tour of minimum weight an \emph{optimal}
TSP tour.
For $\alpha\geq 1$,
an $\alpha$\hyp approximate TSP tour
is a TSP tour whose weight exceeds the minimum weight
by at most a factor~$\alpha$.

\paragraph{Kernelization}
Kernelization is a %
formalization
of data reduction
with provable performance guarantees~\citep{FLSZ19}.

A \emph{parameterized problem}
is a pair
$(\Pi,\kappa)$
where $\Pi\subseteq\{0,1\}^*$ is a decision problem
and $\kappa\colon\{0,1\}^*\to\N$ is 
a polynomial\hyp time computable function
called a \emph{parameterization}.
A~\emph{kernelization} for a parameterized problem~$(\Pi,\kappa)$
is a polynomial\hyp time algorithm
that maps any instance~$x\in\{0,1\}^*$
to an instance~$x'\in\{0,1\}^*$
such that $x\in \Pi \iff x'\in \Pi$ and such that
$|x'|\leq g(\kappa(x))$ for some computable function~\(g\).
We call \(x'\) the \emph{problem kernel}
and \(g\) its \emph{size}.

\section{Data reduction}
\label{sec:dr}
\noindent
Our data reduction rule for \gtsp{} is based
on the following observation:
Let $C\subseteq V$ be a vertex cover of~$G=(V,E)$.
Assume there is a TSP tour~$\tour$ in~$G$.
Since any vertex~$s\in V\setminus C$
has neighbors only in~$C$,
the tour~$\tour$ can traverse~$s$ either via
a subwalk $(u,s,v)$ for $(u,v)\in (C)_2$,
which we call a \emph{hop},
or via a subwalk $(w,s,w)$ for~$w\in C$,
which we call a \emph{loop}.
In the latter case,
one can assume that
$w$~is a vertex~$w\in N(s)\subseteq C$
with minimum~$\w(\{s,w\})$.

\looseness=-1
We will show that
there is an optimal TSP tour~$\tour$
such that,
for each $(u,v)\in (C)_2$,
there is at most one vertex~$s\in V\setminus C$
traversed by~$\tour$ via a hop~$(u,s,v)$.
We will find a superset~$S$ of these vertices
and delete all others,
since they will be traversed using loops of
easily computable costs.
The superset~$S$ will be found by matching vertices in $V\setminus C$
to hops in the following \emph{\hopgt}.

\begin{definition}
  The \emph{\hopgt}~$\hopg$
  for a graph~$G=(V,E)$ with vertex cover~$C\subseteq V$
  is an edge\hyp weighted
  bipartite graph whose vertex set is partitioned into
  two sets
  \begin{align*}
    X&:=(C)_2,\\
    Y&:=V \setminus C.
  \end{align*}
  For each $x = (u, v) \in X$ and $y \in (N(u) \cap N(v)) \setminus C$,
  it contains an edge $\{x, y\}$ of cost
  \begin{align*}
    \wh(\{x,y\}):=\w(\{u, s\}) + \w(\{s, v\}) - 2\w_{min}(s), 
  \end{align*}
  where
  $\w_{\min}(s) := \min_{w \in N(s)} \w(\{s, w\}).$
\end{definition}
\noindent
The edge cost in the hop graph can be interpreted as
a penalty for traversing the vertex~$s$ via a hop~$(u, s, v)$ instead of a loop.

\begin{rrule}
  \label{ourrule}
  Let $(G,\w,W)$ be a \gtsp{} instance.  We obtain a reduced instance~$(G',\w',W')$ as follows.

  \begin{enumerate}
  \item Compute a vertex cover~$C$ with $|C|\leq 2\vc$ for~$G$.
  
  \item Compute the hop graph~$\hopg$ for~$G$ and~$C$.

  \item Compute a maximum\hyp cardinality matching~$M^*$
    of minimal cost with respect to~$\wh$
    in~$\hopg$.
    
  \item Denoting by~$S$ the vertices in $V \setminus C=Y$
  matched in~$M^*$,
  the output \gtsp{}
  instance~$(G',\w',W')$
  is obtained
  by deleting from~$G$
  all vertices (and their incident edges)
  except for~$C\cup S$
  and putting
  \[
    W'=W-2\smashoperator{\sum_{v\in V\setminus (C\cup S)}}\w_{\min}(v).
  \]
  \end{enumerate}
  The new weight function~$\w'$ is the same as $\w$
  on all edges that are not deleted.
\end{rrule}
\noindent

\noindent
All steps work in polynomial time \citep{GJ79,Sch03}.
Moreover,
from any TSP tour~$\tour'$ for~$G'$,
a TSP tour~$P$ for~$G$ can be obtained by adding to~$P'$
a loop~$(w,v,w)$
for each vertex~$v$ present in~$G$
but missing in~$G'$,
where $w\in C$ minimizes $\w(\{w, v\})$.

\section{Problem kernel size analysis}
\label{sec:size}
\noindent
In this section,
we prove that the graph returned
by \cref{ourrule} satisfies
the size bounds stated by \cref{thm}(i--iii).

\begin{lemma}\label{lem:size}
  The graph returned by \cref{ourrule}
  satisfies \cref{thm}\eqref{thm:1},
  \eqref{thm:2}, and \eqref{thm:3}.
\end{lemma}
\begin{proof}
\cref{ourrule}
retains only the vertices in~$C \cup S$.
We know that $|C|\leq2\vc$.
Moreover,
each vertex in~$S$
is matched to some vertex in~$X=(C)_2$,
and thus $|S|\leq|X|=\vc\cdot(\vc-1)$.
Thus,
the remaining graph has at most
\[
  2\vc{} + \vc\cdot(\vc-1)=\vc^2+\vc\quad\text{vertices}.
\]
\looseness=-1
Each remaining edge
has both endpoints in~$C$
or one endpoint in~$C$ and one endpoint in~$S$.
That is,
there are at most
\[
  \binom{2\vc}{2}+  2\vc\cdot\vc\cdot(\vc-1)=2\vc^3-\vc\quad\text{edges}.
\]
Applying Lemma~13 of \citet{BCK+22},
we can reduce the bit\hyp size to
$O((\vc^3)^4)=O(\vc^{12})$.
\end{proof}

\section{Correctness}
\label{sec:correct}
\noindent
To prove the correctness of \cref{ourrule},
we will show that there is an optimal TSP tour~$P$
such that all vertices of~$V\setminus C$
traversed by~$P$ via hops are contained
in the set~$S$, which is not deleted.
Equivalently,
all vertices deleted by
\cref{ourrule} are traversed by~$P$ via loops,
whose costs are easily computable.

To prove this,
in addition to the hop graph~$H$,
which represents all possible hops and their penalties in comparison to loops,
we study the (multi)graph~$H^P$ 
representing the hops 
actually made by~$P$. %

\begin{definition}
  Let~$\hopg^P$ be a multigraph on the same vertex set
  as the hop graph~$H$
  and containing each edge $\{x,y\}$ of~$H$,
  where $x = (u, v) \in X$ and $y \in Y$,
  with a multiplicity
  that equals the number of hops $(u, y, v)$ in the tour~$P$.
\end{definition}

\noindent
The following lemma shows that
we can actually assume
$H^P$~to be an ordinary graph,
and a simple one at that:
every connected component of the graph~$H^P$
forms a ``star'' with
a vertex~$y \in Y$
in the center and only attached to
degree-one nodes in~$X$.

\begin{lemma}
\label{lem:cardinality-of-T}
Let $P$ be an optimal TSP tour minimizing the number of hops.
Then, each vertex in~$X=(C)_2$
has degree at most one in~$H^{P}$.
\end{lemma}
\begin{proof}
  Suppose,
  towards a contradiction,
  that that $P$~contains
  two hops $(u, s_1, v)$ and $(u, s_2, v)$ for $(u, v) \in (C)_2$,
  that is,
  $P$~can be decomposed
  into $\allowbreak P = P_1 \conc (u, s_1, v) \conc P_2 \conc (u, s_2, v) \conc P_3$
  (possibly, $s_1=s_2$).
  Then
  $\tilde{P} = P_1 \conc P_2^R \conc P_3$ is a closed walk
  that visits all vertices of~$G$,
  except, maybe, $s_1$ and $s_2$.

  For $i\in\{1,2\}$,
  we add %
  to~$\tilde{P}$
  a loop~$(w_i, s_i, w_i)$,
  where $w_i\in N(s_i)\subseteq C$
  minimizes $\w(\{w_i, s_i\})$,
  and get a valid TSP tour~$\hat{P}$ with
  \begin{align*}
    \w(\hat{P}) &= \w(\tilde{P}) + 2\w_{{\min}} (s_1) + 2\w_{{\min}} (s_2)\\
                &= \w(P) + 2\w_{{\min}} (s_1) + 2\w_{{\min}} (s_2)\\
                &\phantom{{}=\w(P)}
                  - \w(\{u, s_1\}) - \w(\{v, s_1\})\\
                &\phantom{{}=\w(P)} - \w(\{u, s_2\}) - \w(\{v, s_2\})\leq \w(P).
  \end{align*}
  Thus,
  the cost of $\hat{P}$ is not larger than that of~$P$,
  yet $\hat{P}$ has less hops than~$P$,
  which contradicts the choice of $P$.
\end{proof}

\noindent
We show that there is an optimal TSP tour
such that the set of vertices $S\subseteq V\setminus C$
not deleted by \cref{ourrule}
contains all vertices that~$P$ traverses via hops.

\begin{lemma}
\label{lem:main}
There exists an optimal TSP tour $P^*$ such that $S\subseteq V\setminus C=Y$
contains all vertices that $P^*$ traverses via hops.
\end{lemma}
\begin{proof}
  Let $P$~be an optimal TSP tour
  with a minimum number of hops.
  If there is no vertex~$s_0\in Y\setminus S$
  traversed by~$P$ via a hop,
  then the lemma holds.
  Otherwise,
  consider any path
  in the hop graph~$H$
  of the form
  \begin{align}
  \Bigl(s_0, (u_0, v_0), s_1, (u_1, v_1), s_2, (u_2,v_2),\ldots \Bigr)\label{path}
  \end{align}
  such that $\{s_i, (u_i, v_i)\}$ is an edge of $H^P$
  and $\{(u_i, v_i), s_{i+1}\}$ is an edge of the matching~$M^*$ in~$H$.
  We show that this path has no cycles
  by showing that it enters each vertex at most once.
  \begin{itemize}
  \item The path enters a $s_i$ only via
    a matching edge $\{(u_{j},v_{j}), s_i\}\in M^*$.
    There can be only one matching edge incident to~$s_i$.
  \item No matching edge is incident to $s_0$,
    since $s_0\notin S$,
    and therefore the path never enters~$s_0$ again.
  \item The path enters $(u_i,v_i)$
    only via an edge $\{s_j, (u_i,v_i)\}$ of~$H^P$.
    By \cref{lem:cardinality-of-T},
    there is only one such edge %
    in~$H^P$.
  \end{itemize}
  We can thus consider any path of the form \eqref{path}
  of maximum length.
  Its last edge belongs to the matching~$M^*$:
  otherwise,
  exchanging matching edges and non\hyp matching edges along the path
  would increase the cardinality of~$M^*$,
  contradicting the choice of~$M^*$.
  Thus,
  the path ends at some vertex~$s_k$, which is isolated in~$H^P$,
  and therefore traversed by~$P$ via a loop.
  The path can be partitioned into two disjoint matchings
  \begin{align*}
    M_1&:=\{\{s_i,(u_i,v_i)\}\mid 0\leq i<k\}\text{ and}\\
    M_2&:=\{\{s_{i+1},(u_i,v_i)\}\mid 0\leq i<k\}\subseteq M^*.
  \end{align*}
  Now turn the tour~$P$ into a tour~$P'$ by
  removing the loop visiting~$s_k$,
  for each $i\in\{0,\dots,k-1\}$ replacing the
  hop~$(u_i,s_i,v_i)$ by a hop~$(u_i,s_{i+1},v_i)$,
  and adding a loop~$(w,s_0,w)$
  for the $w\in V \setminus C$ minimizing $\cost(\{s_0,w\})$.
  By the definition of~$\wh$,
  for the path $P'$,
  \[
    \w(P') - \w(P) = \wh(M_2) - \wh(M_1).
  \]
  Notice that, since $M^*$ is of minimum cost with respect to~$\wh$,
  \[
    \wh(M^*) \leq \wh((M \setminus  M_2) \cup M_1) = \wh(M^*) - \wh(M_2) + \wh(M_1),
  \]
  so that $\wh(M_2) \leq \wh(M_1)$. Thus,
  $P'$~is also optimal.
  The number of hops of~$P$ and~$P'$
  is the same,
  yet $P'$~visits strictly less nodes from $Y \setminus S$ via hops,
  namely~$s_0$.
  Thus, we can continue this operation
  until arriving at an optimal TSP tour~$P^*$
  that does not visit nodes from $Y \setminus S$ via hops.
\end{proof}

\noindent
Finally,
we can conclude this work by proving:

\begin{lemma}
  \label{lem:corr}
  \cref{ourrule} is correct.
\end{lemma}
\begin{proof}
  Let $P'$~be a TSP tour for~$G'$.
  Then
  adding to~$P'$ a loop~$(w,s,w)$
  for each vertex~$s$ present in~$G$ but missing in~$G'$,
  where $\w(\{w, s\}) = \w_{\min}(s)$,
  gives a TSP tour~$P$ for~$G$
  with cost
  \[
    \w(P)=\w(P')+2\smashoperator{\sum_{s\in V\setminus (C\cup\fS)}}\w_{\min}(s).
  \]
  In the other direction,
  let $P$~be an optimal TSP tour for~$G$ from \cref{lem:main}. Since vertices in $V \setminus C$ have no edges between them, the only way
  any vertex~$s\in V \setminus (C \cup S)$
  can be visited by~$P$ is via a loop $(w,s,w)$.
  Removing these loops from~$P$ gives a TSP tour~$P'$
  for~$G'$ and its cost is
  \[
    \w(P')= \w(P)-2\smashoperator{\sum_{s\in V\setminus (C\cup\fS)}}\w_{\min}(s).\qedhere
  \]
\end{proof} 

\noindent
\cref{thm} now follows from
\cref{lem:size,lem:corr}.
Regarding \cref{rem:1apx},
the proof of \cref{lem:corr}
shows that \emph{any} solution~$P'$ for the reduced
instance~$(G',\w',W')$ can be turned into a solution~$P$
of cost
$\w(P)=\w'(P')+\Delta$ for the original instance~$(G,\w,W)$, where
\[
  \Delta:=2\smashoperator{\sum_{s\in V\setminus (C\cup\fS)}}\w_{\min}(s).
\]
Thus,
if $P^*$~is an optimal  and
$P'$~is an $\alpha$\hyp approximate solution for~$G'$,
we turn it into a solution~$P$ for~$G$ of cost
\begin{align*}
  \w(P)&=\w'(P')+\Delta\leq \alpha\w'(P^*)+\Delta\leq \alpha(\w'(P^*)+\Delta).
\end{align*}
Note that,
by
\cref{lem:corr},
$\w'(P^*)+\Delta$ is precisely the cost of an optimal solution
for $G$.  Thus, $P$~is $\alpha$\hyp approximate for~$G$ also.

\paragraph{Acknowledgments}
We thank the anonymous referees from \emph{Operations Research Letters} for their valuable comments.

\bibliographystyle{gtsp}
\bibliography{gtsp}

\end{document}